\documentclass[12pt]{amsart}
\usepackage{amsmath,amssymb,amsfonts,amsthm,graphicx,color,fullpage}

\newcommand\bZ{\mathbb Z}
\newcommand\bR{\mathbb R}
\newcommand\bC{\mathbb C}
\newcommand\bP{\mathbb P}
\newcommand\beq{\begin{equation}}
\newcommand\eeq{\end{equation}}

\begin{document}
\title{Multiscale limit for finite-gap Sine-Gordon Solutions and Calculation of Topological 
Charge using Theta-functional Formulae}
\author{P. G. Grinevich}
\address{L.D. Landau Institute for Theoretical Physics, Russian Academy of Sciences}
\email{pgg@landau.ac.ru}
\author{K. V. Kaipa}
\address{Department of Mathematics, University of Maryland, College Park}
\email{kaipa@math.umd.edu}
\thanks{
This work has been partially supported by the Program of Russian
Academy of Sciences ``Fundamental problems of nonlinear dynamics'', 
RFBR grant 09-01-92439-KE-a and by the grant NSh.1824.2008.1 --
Leading Scientific Schools of the Russian Federation President's Council for grants}
\begin{abstract}
In our paper we introduce the so-called multiscale limit for spectral curves, 
associated with real finite-gap Sine-Gordon solutions. This technique allows to 
solve the old problem of calculating the density of topological charge for real 
finite-gap Sine-Gordon solutions directly from the $\theta$-functional formulas.
\end{abstract}
\maketitle
\newtheorem{lemma}{Lemma}
\newtheorem{theorem}{Theorem}
\newtheorem{definition}{Definition}
\newcounter{one}
\setcounter{one}{1}
\newcounter{six}
\setcounter{six}{6}
\newcounter{nineteen}
\setcounter{nineteen}{19}
\renewcommand{\Re}{\mathop{\mathrm{Re}}}
\renewcommand{\Im}{\mathop{\mathrm{Im}}}
\numberwithin{equation}{section}

\section{Introduction}
One of the most important integrable equations of mathematical physics
is the sine-Gordon equation (SG) 
\begin{equation}
u_{tt}-u_{xx}+\sin u=0.
\end{equation}
It appeared in the \Roman{nineteen} century in the theory of constant negative
curvature surfaces in $\bR^3$ in the light-cone representation
$$ 
u_{\xi\eta}=4\sin u. 
$$
Already in the \Roman{nineteen} century it was shown, that the SG equation 
has some remarkable properties, in particular, sufficiently many exact
solutions were constructed using the so-called B\"{a}cklund
transformations. This equation naturally arose in many areas of
mathematics and physics.

The modern approach to the sine-Gordon equation is based on
representation in the form of compatibility conditions for a pair of
auxiliary linear problems. It was first found in \cite{AKNS} by
Ablowitz, Kaup, Newell and Segur. The $\theta$-functional formulas 
for finite-gap SG solutions were obtained by Its, Kozel and Kotlyarov
in \cite{IK}, \cite{KK}, the reality constraints on the spectral
curves were also found in these papers. The reality constraints on the
divisor turned out to be rather nontrivial, and they were found by 
Cherednik \cite{Cher}. In \cite{Cher} it was also shown, that all real
finite-gap SG solutions are automatically non-singular. Divisors, 
generating real solutions are called {\bf admissible}. For a fixed
spectral curve the number of connected components for the variety of
admissible divisors may contain more that one connected components,
the number of these components was also calculated in \cite{Cher}. A
characterization of these components in terms of the Jacobi variety
of the spectral curve was obtained by Dubrovin and Natanzon \cite{DN}, and Ercolani and Forest 
\cite{EF}.
 
It is natural to call SG solutions {\bf periodic}  in $x$
with period $T$ if the quantity $e^{iu}$ is periodic 
with period $T$ in the real variable $x \in \bR$. This definition does not assume 
that the function $u(x,t)$ is periodic in $x$, but
$$
u(x+T,t) = u(x,t)+2\pi n, \ \ n \in \bZ.
$$
The quantity $n$ is called the {\bf topological charge} of $u$. 

For generic real finite-gap SG solutions the function  $e^{iu(x,t)}$ 
is quasiperiodic in $x$, $t$, but $u(x,t)$ is not  quasiperiodic.
It is easy to show, that for real finite-gap SG solutions
constructed by non-degenerate spectral curves the 
{\bf density of topological charge}
\beq  
\label{1.2}
\bar n =\lim\limits_{T\rightarrow\infty}\frac{u(x+T,t)-u(x,t)} {2\pi T} 
\eeq
is well-defined. This density is the most basic and useful characteristic of 
real solutions. In particular it is a conservation law for the SG
hierarchy ``surviving'' generic perturbations of the type 
$$
u_{tt}-u_{xx}+\sin u=\varepsilon F(\sin u, \cos u, u_x, u_t,
u_{xx},u_{xt},u_{tt}, \ldots)
$$

The problem of calculating the topological charge in terms of the
finite-gap spectral data was first raised and partially solved 
in 1982 by Dubrovin and Novikov \cite{DubNov}. 
In \cite{Nov} Novikov pointed out, that the formulas of 
\cite{DubNov} are meaningful only if the spectral 
curve is assumed to be sufficiently close to a degenerate one. A complete solution
was obtained by Grinevich and Novikov only in 2001  \cite{GN}, \cite{GN2}.  \\

The calculation of topological charge density in  \cite{DubNov},
\cite{Nov},  \cite{GN}, \cite{GN2} was based on the so-called 
algebro-topological approach, and explicit $\theta$-functional 
formulas for $u(x,t)$ were not essentially used.  Therefore
a paradoxical situation arose: we have explicit $\theta$-functional formulae 
for the solution, but they do not help to answer the most basic questions 
about the solution. As pointed out by S.P. Novikov, if the $\theta$-functional form of 
solutions is to be considered as an effective one, then it should 
be possible to calculate $\bar n$ directly from the
$\theta$-functional formulae. However, no way to achieve this goal was
found until now. 

In the present article, we suggest a new approach for studying the topological 
characteristics of real finite-gap SG solutions. It is based on the so-called 
{\bf multiscale} or {\bf elliptic  limit}. In  this limit the spectral 
curve $\Gamma$ is deformed to a singular nodal curve having elliptic 
curves as components (and possibly an additional hyperelliptic component).
In particular, we demonstrate, that this approach allows us to extract 
the density of topological charge directly from the $\theta$-functional formulae.

The authors would like to express their gratitude to Professor Novikov 
for attracting their attention to this problem and stimulating discussions.

\section{Complex Finite-gap Sine-Gordon solutions}

We recall the construction of finite-gap SG solutions (\cite{DN},
\cite{GN}, \cite{KK}). The spectral data is a pair $(\Gamma, D)$, 
where $\Gamma$ is a nonsingular hyperelliptic curve with branching points
$\{0, \infty\}$ and $\{E_1,E_2,\cdots,E_{2g}\}$.
\beq
\label{2.1}
\Gamma: \mu^2 = \lambda \prod_{i=1}^{2g} (\lambda - E_i)
\quad 
\mbox{with} \;\;  E_1, \cdots, E_{2g} \;\; \mbox{distinct nonzero complex 
numbers}
\eeq
\noindent and $D=\gamma_1+\gamma_2+\cdots+\gamma_g$ is a divisor of points:
with $\gamma_1, \cdots, \gamma_g \in \Gamma \backslash \{0,\infty \}$. 
For sufficiently small $(x,t)$ (made precise below), the finite-gap solution $u(x,t)$ 
obtained from this spectral data
is non-singular i.e., takes values in the finite complex plane. In order 
to present 
the formula for $u(x,t)$ we need to introduce some notation.
We choose a symplectic basis of cycles $\{a_1, b_1, \cdots, a_g, b_g\}$ in $H_1(\Gamma, \bZ)$ 
and also a basis for the 
$g$-dimensional space of holomorphic differentials on $\Gamma$: 
$\vec{\omega} = (\omega_1, \cdots,\omega_g)$
normalized such that $\int_{a_j} \omega_i =\delta_{i j}$. The Riemann matrix 
of $\Gamma$ is the matrix defined by $B_{ij}  = 
\int_{b_j} \omega_i$. The matrix $B$ is symmetric and its imaginary part is positive definite. 
The Riemann theta function associated with 
$B$ is the entire function of $g$ complex variables $z=(z_1, \cdots, z_g)$ defined by
\beq 
\label{2.2}
\theta(z|B) = \sum_{ n \in \bZ^g} \exp( \pi i \,n^t B n ) \, \exp(2 \pi i \,n^t z) 
\eeq
\noindent It satisfies the transformation rule
\beq
\label{2.3}
\theta(z+N+BM) = \theta(z) \, \exp \left( - \pi i \, \left[ 2 M^t z + M^t B M \right]  \right) 
\qquad   {\rm where} \quad  N, M  \in \bZ^g
\eeq
The complex torus $J(\Gamma)=\bC^g/\{\bZ^g + B \bZ^g\}$ is the Jacobian variety 
of $\Gamma$, and the Abel map $A:\Gamma \to J(\Gamma)$ 
is defined by $P \mapsto \int_{\infty}^P \vec{\omega}$. 
Let $K \in \bC^g$ be the associated vector of Riemann constants. 
Let $\omega_{\infty}$ and $\omega_0$ be abelian 
differentials of the second kind having  double poles at $P=\infty$ and $P=0$ 
respectively with the principal parts being $ - \lambda \, d (1/\sqrt{\lambda})$ 
and $ -1/\lambda \, d \sqrt{\lambda}$ respectively and normalized to have zero $a$-periods. 
(where by $\sqrt \lambda$ and  $1/ \sqrt \lambda$ we mean  local coordinates $k_0$ 
and $ k_{\infty}$ at $P=0$ and $P=\infty$ respectively with 
$k_0^2=\lambda$ and  $k_{\infty}^2=1/\lambda$). We define two vectors 
$U, V \in \bC^g$ by 
\beq 
\label{2.4}
U_j = \frac{1}{2 \pi} \int_{b_j} \omega_{\infty}, \qquad  V_j = \frac{1}{2 \pi } \int_{b_j} \omega_{0} 
\eeq
The $\theta$-functional formula for $u(x,t)$ is given by (\cite{DN}, \cite{KK}):
\beq
 \label{2.5}
\begin{split} 
e^{i u(x,t)} = C_1 \frac{ \theta( A(0)+ z(x,t))\, \theta( -A(0)+ z(x,t)) }{ \theta^2(z(x,t)) } \\
 z(x,t) =  -A(D)+  x (V-U)/4 -  t (U+V)/4 - K
\end{split}
\eeq 
where $C_1^2 =  \exp( \pi  i \, \epsilon^t B \epsilon /2 )$. We remark that $A(0)$ is a half-period 
(because $A(P)-A(Q)$ is always a half-period whenever $P, Q$ are branch points of a $2$-sheeted 
cover $\Gamma \to \bP^1$).  Let $\Theta=\{z \in J(\Gamma)\,|\, 
\theta(z)=0\}$ (which is well-defined by virtue of (\ref{2.3})) denote the $\theta$-divisor. Let 
$\mathcal U \subset \bR^2$ be a neighborhood of $(0,0)$ in the $(x,t)$ plane for which the 
image in $J(\Gamma)$ of $z(x,t)$ is disjoint from the divisor $\Theta \cup (\Theta+A(0))$. 
If $\mathcal U$ is simply-connected, then (\ref{2.5}) defines the function $u(x,t)$ uniquely 
once $u(0,0)$ is specified. In the real case  $\mathcal U$ may be chosen to coincide with
the whole  $\bR^2$ (see below). 

\section{ Real Tori and Real Solutions }
\label{real_tori}
In order to obtain real solutions  $u(x,t)$, we must impose some conditions on the 
spectral data. The reality condition on $\Gamma$ (found in \cite{IK}, \cite{KK}) is:
\beq
\label{3.1}
\{\overline{E_1}, \cdots, \overline{E_{2g}}\} =  \{ E_1, \cdots,  E_{2g}\} \quad \mbox{and} 
\quad  E_i \in \bR \Rightarrow E_i < 0
\eeq
Let $2m$ denote the number of real $E_i$. We order the real branch points as $0>E_1>E_2> \cdots>E_{2m}$ 
and also assume 
$E_{2i} = \overline{E_{2i-1}}$ for $m+1 \leq i \leq g$. The reality condition on the divisor found by 
Cherednik \cite{Cher} is
\beq
\label{3.2} 
D+ \tau D -0 -\infty  \sim  \mathcal{K} 
\eeq 
\noindent where $\tau$ is the anti-holomorphic involution of $\Gamma$ given by 
$\tau:(\lambda, \mu) \mapsto (\bar{\lambda}, \bar{\mu})$ and 
$\mathcal{K}$ is the canonical class, and $\sim$ denotes the relation of linear 
equivalence of divisors. Such a divisor will be called  {\bf admissible}. 
The Abel map $A:\Gamma \to J(\Gamma)$ extends to divisors by 
$A(\sum_i P_i - \sum_j Q_j) = \sum_i A(P_i) - \sum_j A(Q_j)$.
It was shown in the works \cite{Cher}, \cite{DN}, \cite{EF} that the image in 
$J(\Gamma)$ of the admissible divisors under the Abel map, consists of 
$2^m$ components each of which is a real $g$-dimensional torus. We use here a 
concrete description of these tori, similar to the one in \cite{DN}.

We choose a special basis of cycles $\{a_1, b_1, \cdots, a_g, b_g\}$  suggested in \cite{DubNov}, 
and depicted in Fig.~2 (where the parameter $k$ must be taken to 
be $1$). The picture shows the $\lambda$-plane with the thick-dashed lines 
representing the system of cuts, and the transition from solid to dashed 
line in the cycles $\{b_{m+1}, \cdots, b_{g}\}$ indicating change of sheet across a cut. 
The action of $\tau$ on $H_1(\Gamma, \bZ)$ is given by 
\begin{align*}
\tau a_i &= - a_i \qquad  1 \leq i \leq g \\
\tau b_i &= b_i \quad \qquad 1 \leq i \leq m \\
\tau b_i &=b_i + a_i \quad  m+1 \leq i \leq g 
\end{align*}

\noindent This immediately implies that the effect of $\tau$ on the holomorphic differentials 
is given by $-\omega_j =  \overline{ \tau^* \omega_j}$.
Therefore
\begin{gather}
  \label{3.3}
A(\tau P) = -\overline{A(P)}  \\
  \label{3.4}
\quad -\bar{B}=B +  \bigl( \begin{smallmatrix}   0 & 0\\   0 & I_{g-m} \end{smallmatrix} \bigr)
 \quad \mathrm{or}  \quad
\Re(B)= -1/2 \,  \bigl( \begin{smallmatrix}   0 & 0\\   0 & I_{g-m} \end{smallmatrix} \bigr) 
\end{gather}
where $I_{g-m}$ is the identity matrix of size $g-m$. (From here on, all vectors $v$ written 
as $(v_1, v_2)^t$ are understood to be split into blocks of length $m$ and $g-m$).

In order to describe the real tori, we need to calculate $A(\mathcal K)$ and $A(0)$. 
The divisor $ ( d \lambda/ \mu)  = (2g-2)\cdot \infty$ is canonical hence we obtain 
$A(\mathcal K)=0$. Let $2 \pi i \, \alpha$ and $2 \pi i \, \beta$ denote the vectors of $a$ and  $b$-periods 
of the differential $\omega =\tfrac{1}{2}\,d \log \lambda$. From the bilinear 
relations of Riemann applied to the pair of differentials $\omega, \omega_i$ for 
$1 \leq i \leq g$, we immediately obtain $ (\beta - B \alpha)/2 = A(0)$.
The numbers  $\alpha_j$ and $\beta_j$ are the winding numbers of the projection 
$\lambda(a_j)$ and $\lambda(b_j)$ around $\lambda = 0$, and are
read off from Fig.~2 to be $\alpha =  (1,0)^t$ and $\beta = (0,1)^t$. 
Thus we obtain:
\begin{gather}
  \label{3.5}
A(\mathcal K)=0 \\
  \label{3.6}
A(0) = \frac{1}{2} \, \epsilon' + \frac{1}{2} \,B \epsilon, \quad 
\epsilon = { 1 \choose 0} \quad \epsilon' = { 0 \choose 1}
\end{gather}
Every $z \in J(\Gamma)$ can be written as $z = x + B\,y$ for unique vectors 
$x, y \in \bR^g/\bZ^g$. We will denote the subset of elements of 
the form $\{x + B\, 0\}$ as $T^g$. Clearly $T^g \subset J(\Gamma)$ is isomorphic to 
$\bR^g/\bZ^g$. Applying the Abel map to the Cherednik 
reality condition  $D+ \tau D -0 -\infty  \sim \mathcal{K}$ and using  
(\ref{3.3}), (\ref{3.5}) and (\ref{3.6}) we obtain:
 \beq 
\label{3.7}
\begin{split}  A(D) &=  x + B \, { s/4 \choose  1/2}, \quad
  s={ {s_1\atop s_2} \choose  {\cdot \atop \,s_m}} \\
{\mbox{ for some}} \quad x \in &T^g \quad \mbox{and} \quad  s_1,\ldots,s_m =\pm1 \end{split} 
\eeq
For each of the $2^m$ collection of symbols  $s_1,\ldots,s_m=\pm 1$, let $T_s$ denote the subset of $J(\Gamma)$ 
consisting of elements of the form  $-A(D)-K$ where $A(D)$ is as in (\ref{3.7}) and $K$ is the vector of 
Riemann constants.
\beq 
\label{3.8}
T_s = -T^g  - B \, { s/4 \choose  1/2} -K  
\eeq
The symbol $s$ will be called the {\bf topological type} of the admissible divisor $D$.  
In the sequel, the terminology real tori will refer to the $T_s$. 
We collect the essential facts (from \cite{Cher}) about the real tori in the following lemma. 
\begin{lemma}
i) The real tori $T_s$ do not intersect the divisor $\Theta \cup (\Theta+A(0))$\\ 
ii) Let $D$ be an admissible divisor of topological type $s$, and let  
$ z(x,t) =  -A(D)+  x (V-U)/4 -  t (U+V)/4 - K$. Then $z(x,t) \in 
T_s$ for all $x,t$. \\
iii) The real solutions $u(x,t)$ are non-singular for all $x,t$
\end{lemma}
\begin{proof}
{\em i}) We use the fact (\cite{FK} pp 310) that $\Theta$ consists of elements of the 
form $A(\gamma_1+\cdots+\gamma_{g-1})+K$. From this fact (together with 
$\Theta=-\Theta$) it follows that  for a divisor $D$, the condition 
$-A(D)-K \in \Theta \cup (\Theta+A(0))$ is equivalent to 
$D \sim \gamma_1+\cdots+\gamma_g$ for some $\gamma_i$ with $\gamma_g \in \{0,\infty\} $. 
Suppose such a divisor $D$ is also admissible, i.e. $D + \tau D - 0 - \infty \sim \mathcal K$. 
Since the Cherednik condition depends only on the divisor class of $D$, 
we may assume $D=\gamma_1+\cdots+\gamma_{g-1}+\gamma_g$ with $\gamma_g \in \{0,\infty\}$. 
Let $D_1 =  \gamma_1+\gamma_2+\cdots+\gamma_{g-1}$.  The admissibility condition can be rewritten as
\beq  
\label{3.9}
D_1 + \tau D_1 \pm (0-\infty) = (\alpha) \quad \mbox{for some abelian differential} \; \alpha 
\eeq 
If the differential $\alpha$ is holomorphic then (\ref{3.9}) implies that it 
has an odd number of zeros at $\infty$. However every holomorphic differential on $\Gamma$ 
has an even number of zeroes at $\infty$. Therefore $\alpha$ 
is not holomorphic and (\ref{3.9}) now implies that $\alpha$ has a single pole, contradicting 
the residue theorem.  Thus  $T_s$ does not intersect $ \Theta \cup (\Theta + A(0))$.

{\em ii}) the second assertion follows from the first as soon as we show that the vectors $U, V$ 
defined by (\ref{2.4}) are purely real. It is easy to see from the definition of 
$\omega_0, \omega_\infty$ 
that they are fixed by the involution $\omega \mapsto \overline{\tau^* \omega }$.
Equation (\ref{2.4}) together with the fact that $\tau b_i = b_i \; \rm{mod} \; a_i$, now 
implies that  $U$ and $V$ are purely real.

 {\em iii}) The formula (\ref{2.5}) for  $u(x,t)$ shows that it is nonsingular as long 
as $ z(x,t) \notin  \Theta \cup (\Theta+A(0))$. Thus the first two
assertions imply the third one.
\end{proof}

\noindent To conclude this section, we calculate the vector of Riemann constants $K$ for the 
chosen base point of Abel map ($P=\infty$), and the choice
of basic cycles. A formula for $K$ is well-known (\cite{FK}, pp 325) when the 
base point of Abel map is a branch point of a $2$-sheeted cover
$\lambda:\Gamma \to \bP^1$. It equals $\sum_i A(P_i)$ where the sum is over 
those branch points $P_i$ for which $A(P_i)$ is an
odd half-period (i.e. of the form $(n+Bm)/2$ with the scalar product $(n, m)$ 
being an odd integer). In the present case, we obtain:
\beq \label{3.10} K = A(E_1)+A(E_3)+ \cdots +A(E_{2g-1}) \eeq
\noindent Let $2 \pi i \, \alpha'$ and $2 \pi i \, \beta'$ be the vectors of $a$ and $ b$-periods 
of the differential $ \tilde \omega = \tfrac{1}{2}\, d \log
\prod_{i=1}^g(\lambda-E_{2i-1})$, which is a differential of the third kind 
having residues $+1$ at $E_1, E_3, \cdots, E_{2g-1}$ and residue $-g$ at
$\infty$. From the bilinear relations of Riemann applied to the pair of 
differentials $\omega_i, \tilde \omega$ for $1
\leq i \leq g$, we obtain
\[ (\beta' - B \alpha')/2 =  A(E_1)+A(E_3)+ \cdots +A(E_{2g-1})  = K \]
The numbers  $\alpha_i'$ and  $ \beta_i'$ being  $ \tfrac{1}{2 \pi  i}\int_{a_i} \tilde \omega$ and  
$\tfrac{1}{2 \pi  i} \int_{b_i} \tilde \omega$ respectively are equal to the
sum of the winding numbers about $E_1, E_3, \cdots, E_{2g-1}$ of the projected 
curves $\lambda(a_i)$ and $\lambda(b_i)$ respectively. The winding
numbers can be read off from Fig.~2. We obtain:
\beq 
\label{3.11}
K= \frac{1}{2} \,{ 1 \choose \nu_2} + \frac{1}{2} \, B \,{ \nu_1 \choose 1 } 
\quad \mathrm{  where } \; {\nu_1 \choose \nu_2}=
{ {1\atop 2} \choose  {\cdot \atop \,g}}
\eeq

\section{Basic charges and Topological Charge in two special cases}
\label{limiting_charges}
In order to calculate the topological charge density $\bar n$ defined in (\ref{1.2}) 
for the finite-gap solution $u(x,t)$ (\ref{2.5}), following \cite{GN} we introduce 
some integer quantities $(n_1, n_2, \cdots, n_g)$ associated with each real torus $T_s$. 
We call the quantities 
$n_j$ as {\bf basic charges}. Let $e_j$ denote the $j$-th standard basis vector of 
$\bR^g$ and let $\{ z_j(T) =  -A(D) - T e_j - K \, | \, 0 \leq T \leq 1\}$ denote 
the $j$-th basic cycle of the real torus $T_s$. We define the basic charges $n_j$ by:
\beq
\label{4.1}
 n_j = \frac{1}{2 \pi i} \int_{T=0}^{T=1} d \log e^{i u_j(T)}, \; \mbox{where}\quad   e^{i u_j(T)} = C_1 \frac{ \theta( A(0)+ z_j(T))\, \theta( -A(0)+ 
z_j(T)) }{ \theta^2(z_j(T)) }
\eeq
The density of topological charge $\bar n$ for the solution real finite-gap solution $u(x,t)$ 
depends only on its topological type, and is related to 
the basic charges by (see \cite{GN2}):
\beq 
\label{4.2}
\bar{n}= \sum_{j=1}^g (U_j-V_j) \, n_j /4 
\eeq
It will be convenient to replace the quantity $\epsilon=(1,0)^t$ in the formula (\ref{3.6})  
for $A(0)$ by the quantity $\tilde \epsilon$ defined by 
$\tilde \epsilon_j = (-1)^j s_j$ for $j \leq m$ and $\tilde \epsilon_j = 0$ for $j > m$. 
Using this new expression for $A(0)$ in equation (\ref{4.1}) and 
using the $\theta$-transformation rule (\ref{2.3}) we obtain the following formula for $n_j$:
\beq 
\label{4.3}
n_j = -\tilde{\epsilon}_j +  \frac{2}{2 \pi i} \int_{T=0}^{T=1} 
d \log \left(  \frac{\theta( B \tilde{\epsilon}/2 -T e_j -K - B \, { s/4 \choose 1/2} 
)}{ \theta(  -T e_j -K - B \, { s/4 \choose  1/2} ) }  \right) 
\eeq
We calculate the basic charges $n_j$ for two special cases of spectral curves 
$\Gamma$. In the first case $m=0$, i.e. none of the $E_i$ are real.
In this case all components of the quantity $\tilde \epsilon$ defined above are zero. 
Thus both terms in (\ref{4.3}) vanish and we obtain $n_j=0$. 
In the second case we consider $g=m=1$, i.e. elliptic curves with all branching points real. 
We denote the Riemann matrix $B$ by $\tau$. For this case, (\ref{4.3}) can be rewritten as:
\beq
\label{4.4}
n_1 = s_1 +  \frac{2}{2 \pi i} \int_{T=0}^{T=1} d \log 
\left(  \frac{\theta( -T  - (1+\tau)/2  - s_1 \tau /4 -s_1 \tau/2)}
{ \theta(  -T -(1+\tau)/2  -  s_1\tau/4) }  \right) 
\eeq
The integral term is equal to $2 s_1$ times the number of zeroes of $\theta(z|\tau)$ in the region  
$R  = \{ z \in \bC \, |\,  -\tau/4 \leq \Im(z) \leq 
\tau/4, 
-1 \leq \Re(z) \leq 0 \}$. Since the elliptic theta function $\theta(z|\tau)$ vanishes 
if and only if $z \sim (1+\tau)/2$ and the region
$R$ does not contain any such $z$, it follows that the integral term in (\ref{4.4}) 
is zero and we obtain $n_1 = s_1$. For use later in SectionÅ~\ref{calc_top_charge} , we will 
calculate a related quantity $\tilde n_1$ given by:
\beq  
\label{4.5}
\tilde n_1 = -s_1 +  \frac{2}{2 \pi i} \int_{T=0}^{T=1} d \log \left(  \frac{\theta( -T  - 1/2  - s_1 \tau /4 + s_1 \tau/2)}
{ \theta(  -T -1/2  -  s_1\tau/4) }  \right) 
\eeq
In this case the integral term is equal to $-2 s_1$ times the number of zeroes  $\theta(z|\tau)$ in the region  $R$. Hence the integral 
term drops out and we obtain $\tilde n_1 = -s_1$.

\section{Multiscale Limit}
In this section we construct a deformation of the spectral curve $\Gamma$ (\ref{2.1}) 
satisfying the reality condition (\ref{3.1}). For each $k \in [1,\infty)$, 
let $\Gamma(k)$ be the real hyperelliptic curve:

\beq  \label{5.1} \Gamma(k): \mu^2 = \lambda \prod_{i=1}^m ( \lambda - k^{i-1} E_{2i-1})  
( \lambda - k^{i-1} E_{2i}) \, \prod_{i=2m+1}^{2g}  ( \lambda - k^m
E_i)\eeq
We note that the original spectral curve $\Gamma = \Gamma(1)$. The basic cycles and the system of cuts on 
$\Gamma(k)$ are shown in Fig.~2. If
$\Pi_j(k)$ for $m+1 \leq j \leq g$ denotes the cut on the
$\lambda$-plane joining the complex conjugate branch points $k^{m} E_{2j-1}$ and  $k^{m} E_{2j}$, 
then we require $\Pi_j(k) = k^{m} \Pi_j(1)$.
The remaining cuts lie on the negative real line of the $\lambda$-plane and are obvious from Fig.~2.
We will use the notation $\Gamma(k)^+$ to denote the sheet for which $\mu >0$ when $\lambda>0$. As indicated 
in Fig.~2, the cycles $a_j(k), b_j(k)$ 
are completely determined by their projections $\lambda(a_j(k)), \lambda(b_j(k))$ to
the $\lambda$-plane. We specify the latter by:
\beq \label{5.2}
\begin{split}
\lambda_j(a_j(k)) = \lambda(a_j(1)),& \quad  \lambda_j(b_j(k)) = \lambda(b_j(1)) \\
\quad \mbox{where} \;\, \lambda_j = \begin{Bmatrix}
\lambda/k^{j-1} \qquad \mathrm{ if } \quad j \leq m \\
\lambda/k^m \; \; \qquad \mathrm{ if } \quad j > m  \end{Bmatrix}
 \mbox{ and }& \lambda: \Gamma(k) \to \bP^1 \; 
\mbox{are projections} \end{split} 
\eeq
\noindent  where we have introduced rescaled versions $\lambda_j$ of the coordinate function $\lambda$. We also 
associate with the deformation $\Gamma(k)$, $m$ elliptic curves $\mathcal C_1, \cdots, \mathcal C_m$ and a
hyperelliptic curve $\mathcal C_{m+1}$.
\beq \label{5.3}
\begin{split}  \mathcal C_j: y_j^2 = P_j(x_j) = x_j (x_j-E_{2j-1}) (x_j - E_{2j}),  \quad 1 \leq j \leq m \\
  \mathcal C_{m+1}: y_{m+1}^2 = P_{m+1}(x_{m+1}) = x_{m+1} \prod_{i=2m+1}^{2g} (x_{m+1}-E_i). \end{split} 
\eeq
\noindent The system of cuts on these curves is shown in Fig.~1. The cuts decompose
 each $\mathcal C_j$ into two sheets. By the sheet $\mathcal C_j^+$ we will mean the sheet 
for which $y_j >0$ when $x_j>0$. \\
 \begin{figure}[htb]
\input{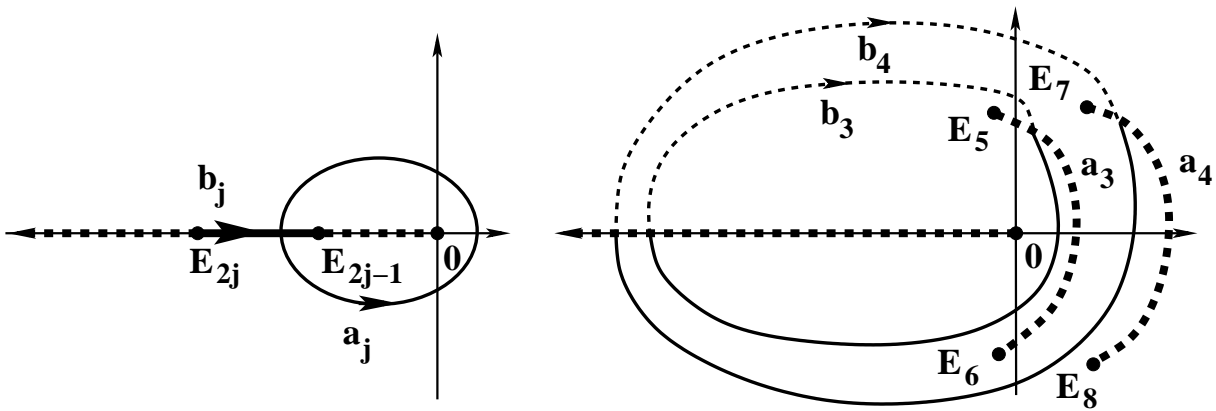}
  \caption{ Basic cycles and cuts on $\mathcal C_j$, $1 \leq j \leq m$ and $\mathcal C_{m+1}$ }
\end{figure}
 \begin{figure}[htb]
\input{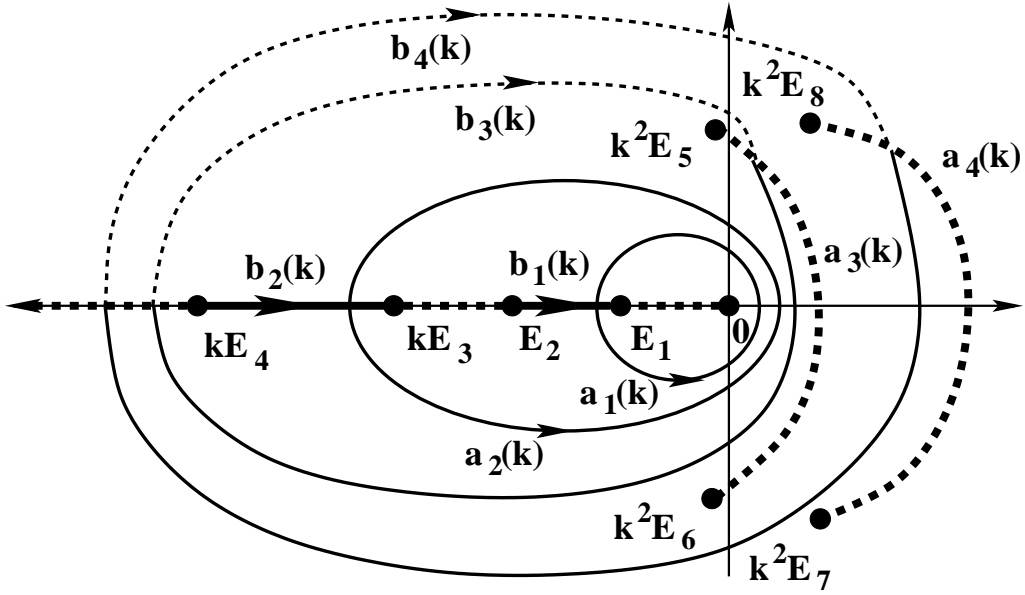}
  \caption{  Basic cycles and cuts on $\Gamma(k)$ for $g=4, m=2$}
\end{figure}

Let $B(k)$ denote the Riemann matrix of $\Gamma(k)$ with respect to this choice of basic cycles. 
We will be concerned with the limit of the pair $(\Gamma(k), B(k))$ as $k \to \infty$.
To this end, we decompose $\Gamma(k)$ into $m+1$ open sets $R_j$, and we also introduce certain 
rescaled versions of the coordinate function $\mu$ on $\Gamma(k)$:
\beq 
\label{5.4}
\begin{split} R_1 &= \{ (\lambda, \mu) \in \Gamma(k) \,|\; 0 \leq | \lambda_1 | < k^{2/3} \}\\
R_j &= \{ (\lambda, \mu) \in \Gamma(k) \,|\; k^{-1/2} < | \lambda_j | < k^{2/3} \}, \quad \mbox{for} \;\, 2 \leq j \leq m \\
R_{m+1} &= \{ (\lambda, \mu) \in \Gamma(k) \,|\; k^{-1/2} < | \lambda_{m+1} | \leq \infty \} \end{split} 
\eeq
\beq 
\label{5.5}
\mu_j = 
\begin{Bmatrix} 
\mu/ \Bigl( \sqrt{k}^{3(j-1)}\, k^{(j+\cdots+m-1)} \, k^{m(g-m)}\, \lambda^{j-1}\, 
(\prod_{l=2j+1}^{2g} E_l)^{0.5} \, \Bigr) 
 \qquad \mathrm{ if } \quad 1 \leq j \leq m \\ 
\mu/ \Bigl( \sqrt{k}\, k^{m(g-m)}\, \lambda^m \, \Bigr)  \qquad \qquad \mathrm{ if } \quad j = m+1
 \end{Bmatrix} 
\eeq
where all square-roots occurring in (\ref{5.5}) are positive real numbers. 
We define  maps $\phi(k) = (\phi_1(k), \phi_2(k), \cdots, 
\phi_{m+1}(k))$  from $\Gamma(k)$  to $(\bP^2)^{m+1}$  given by: 
\beq
\label{5.6}
\phi_j(k) (\lambda, \mu) = (\lambda_j: \mu_j :1) \eeq
We consider the restriction of the maps  $\phi_j(k)$ to the regions $R_i$. Using  (\ref{5.4}) 
and (\ref{5.5}) in (\ref{5.6}) we obtain:
\beq 
\label{5.7}
 \phi_j(k)_{|R_i} (\lambda, \mu) = \begin{Bmatrix} 
(o(1/k):o(1/k):1) \qquad \qquad \quad \mbox{if} \quad i < j \\
(o(1/k):1:o(1/k)) \qquad \qquad \quad \mbox{if} \quad i > j \\
(\lambda_j:\sqrt{P_j(\lambda_j)}\,(1+o(1/k)):1) \quad \mbox{if} \quad i = j  \end{Bmatrix}  
\quad \mbox{as} \quad k \to \infty
\eeq
where, for $(\lambda, \mu) \in \Gamma(k)^+$ the expression 
$(\lambda_j, \sqrt{P_j(\lambda_j)}) \in \mathcal C_j^+$. This can be seen by noting that, 
for $j \leq m$,  ${ \rm sgn}(-\mu) = {\rm sgn}(\lambda^{j-1}) = (-1)^{j-1}$ 
on  $\Gamma(k)^+$ over $( E_{2j}, E_{2j-1})$, and therefore (\ref{5.5}) 
implies that ${\rm sgn}(\mu_j) = -1$ on $\Gamma(k)^+$ over $( E_{2j}, E_{2j-1})$. 
Also ${\rm sgn}(\sqrt{P_j(x_j)}) = -1$ on $\mathcal C_j^+$ 
over $x_j \in ( E_{2j}, E_{2j-1})$. Similarly if $j=m+1$, then both quantities 
$\mu, \lambda^m$ and therefore $\mu_{m+1}$ are positive on $ 
\Gamma(k)^+$ over the positive real 
axis, and $\sqrt{P_{m+1}(x_{m+1})}$ is also positive on $\mathcal C_{m+1}^+$ over the  positive real axis.

We next consider embeddings $ \phi_i: \mathcal C_i \to (\bP^2)^{m+1}$ given by:
\[ \phi_i: (x_i, y_i) \mapsto \{(0:1:0)\}^{i-1} \times (x_j:y_j:1) \times \{(0:0:1)\}^{m+1-i} \]
($\phi_{m+1}$ is  singular  at $\infty \in \mathcal C_{m+1}$ if $g-m>1$, but this is inessential for us). 
Let $\mathcal C \subset(\bP^2)^{m+1}$ 
be defined as $\cup_{i=1}^{m+1} \phi_i(\mathcal C_i)$. Clearly $\mathcal C$ is a nodal hyperelliptic 
curve of genus $g$ having  $m$ double points
$\{(0:1:0)\}^{i} \times \{(0:0:1)\}^{m+1-i}$ for $1 \leq i \leq m$. It is also clear from (\ref{5.7}) that:
  \begin{gather*} \lim_{k \to \infty} \phi(k) ( \Gamma(k) ) \, = \,  \mathcal C  \subset  (\bP^2)^{m+1} \\
\mbox{with} \quad  \lim_{k \to \infty} \phi(k) (R_i) = \phi_i(\mathcal C_i) \subset \mathcal C \\
 \mbox{and} \quad \lim_{k \to \infty} \phi(k)(R_i \cap R_{i+1}) = \{(0:1:0)\}^{i} \times \{(0:0:1)\}^{m+1-i} \quad \mbox{(the $m$ double points)} 
\end{gather*}  

Let $a_j, b_j$ denote the basic cycles on $\phi_j(\mathcal C_j) \subset \mathcal C$ 
( or $\phi_{m+1}(\mathcal C_{m+1}) \subset \mathcal C$ if $j>m$). 
Let $\omega_j$ be the holomorphic differential on the elliptic curve $\mathcal C_j$ (for $j \leq m$) 
satisfying $\int_{a_j} \omega_j = 1$, and define $\tau_j = \int_{b_j} \omega_j$. We define $B_1$ 
to be the diagonal matrix diag$(\tau_1, \cdots, \tau_m)$. Similarly, let $\omega_{m+j}$ for $1 
\leq j \leq g-m$ be holomorphic differentials on the hyperelliptic curve 
$\mathcal C_{m+1}$ satisfying $\int_{a_{m+j}} \omega_{m+i} = \delta_{i j}$ and 
let $(B_2)_{ij} = \int_{b_{m+j}} \omega_{m+i}$  for $1 \leq i,j \leq g-m$. Define $B_{\infty}$ 
to be the block-diagonal matrix $B_{\infty} = $ diag$(B_1, B_2)$. It is the Riemann matrix of 
$\mathcal C$ with respect to the basic cycles $\{a_1, b_1, \cdots, a_g, b_g\}$. 
Using (\ref{5.7}) again, it follows that the component of $ \lim_{k \to \infty} \phi(k) (a_j(k))$ in 
$\phi_i(\mathcal C_i)$ is homologous to $a_i$ if $i=j$ and homologous to zero if $i \neq j$. 
Similarly the component of $ \lim_{k \to \infty} 
\phi(k) (b_j(k))$  in $\phi_i(\mathcal C_i)$ is homologous to $b_i$ if $i=j$ and homologous to zero if $i \neq j$.
\begin{gather*} \lim_{k \to \infty} \phi(k)_* (a_j(k)) = a_j, \quad \mbox{and} 
\quad  \lim_{k \to \infty} \phi(k)_* (b_j(k)) = b_j \quad \\
\mbox{therefore} \quad \lim_{k \to \infty} B(k) = B_{\infty} \end{gather*}

\section{Calculation of Topological Charge}
\label{calc_top_charge}
We reduce the calculation of the basic charges $n_j$ in the general case to the two special cases 
computed in Section~\ref{limiting_charges}. As proved in Lemma 1, the real tori $T_s$ do not intersect the divisor 
$\Theta \cup (\Theta+A(0))$. Therefore the integral term in (\ref{4.3}) stays nonsingular through 
the deformation $B(k)$. It is also constant because it is continuous in the deformation parameter 
$k$ and it is integer valued. In other words the basic charges $n_j$ may be calculated from (\ref{4.3}) 
using $B_{\infty}$ in place of $B$. Also, it follows from the definition 
(\ref{2.2}) that $\theta((z_1, z_2)^t \,|\, {\rm diag}(B_1, B_2)) = \theta(z_1\,|\,B_1)\, 
\theta(z_2\,|\,B_2)$. Using this and the computation (\ref{3.11}) for 
$K$ in the formula (\ref{4.3}), we obtain
\beq
\label{6.1}
n_j = \begin{Bmatrix} s_j +  \frac{2}{2 \pi i} \int_{T=0}^{T=1} d \log \left(  \frac{\theta( -T  - (1+\tau_j)/2  - s_j \tau_j /4 -s_j \tau_j/2)}
{ \theta(  -T -(1+\tau_j)/2  -  s_j\tau_j/4) }  \right)   = s_j \quad \mbox{ if $j \leq m $ is odd} \\
 -s_j +  \frac{2}{2 \pi i} \int_{T=0}^{T=1} d \log \left(  \frac{\theta( -T  - 1/2  - s_j \tau_j /4 +s_j \tau_j/2)}
{ \theta(  -T -1/2  -  s_j\tau_j/4) }  \right)   = -s_j \quad \mbox{ if $j \leq m$ is even} \\
0 \qquad  \mbox{ if $j > m$ } 
\end{Bmatrix} \eeq

Thus we have proved:
\begin{theorem} The topological charge density $\bar n$ for the real finite-gap solution $u(x,t)$ 
for the spectral data $(\Gamma, D)$ with $D$ corresponding to the real torus $T_s$ is given by:
\[ \bar{n}= \sum_{j=1}^g (U_j-V_j) \, n_j /4  \]
where the basic charges $n_j$ are:
\beq 
\label{6.2}
n_j =
\begin{Bmatrix} (-1)^{j-1} s_j \qquad \mathrm{ if } \quad 1 \leq j \leq m \\ 
0 \qquad \qquad \mathrm{ if } \quad j > m \end{Bmatrix} 
\eeq
\end{theorem}

\noindent {\it Remarks}
\begin{enumerate} 
\item In \cite{GN2}, the admissible divisors were characterized by certain symbols 
$\{s_1', \cdots, s_m'\} \in \{ \pm 1\}^m$ defined as follows. Given an admissible divisor 
$D=\{ (\lambda_i, \mu_i) \,|\, 1 \leq j \leq g \}$ let $P(\lambda)$ be the unique polynomial of
degree $g-1$ interpolating the $g$ points $(\lambda_i, \mu_i/\lambda_i)$. Then $P(\lambda)$ is real 
and $s_j'$ is defined to be the sign of $P(\lambda)$ over $[E_{2j}, E_{2j-1}]$. 
It was shown in \cite{GN} that the charges $n_j$ are equal to $(-1)^{j-1} s_j'$ for $j \leq m$ and 
$n_j=$ for $j > m$. Comparing with formula (\ref{6.2}), it follows that the symbols 
$s_j'$ and $s_j$ coincide.
\item The multiscale limit of the spectral curve constructed above was used only for a topological argument. 
The sine-Gordon solutions $u(x,t,k)$ associated with the spectral curve $\Gamma(k)$ 
(and admissible divisors $D(k)$) depend on the vectors $U(k)$ and $V(k)$ mentioned in the
Section~\ref{real_tori} . As $k \to \infty$, some component of $U(k)$ will diverge 
to $\infty$. Thus there is no limiting solution. However asymptotic expansion
in the parameter $k$ of $u(x,t,k)$ involving elliptic (genus $1$) solutions can be written. 
This will be investigated in a future work.
\end{enumerate}

\end{document}